\documentclass[a4paper]{amsart}
\usepackage{amsmath, amsthm, amsfonts, amscd, amssymb}
\usepackage{stmaryrd} 
\usepackage{hyperref}

\hypersetup{
    colorlinks=false,       
    linkcolor=black,          
    linktocpage,             
    citecolor=black,        
    filecolor=black,      
    urlcolor=black          
    linkbordercolor={1 0 0} 	
    citebordercolor={0 1 0} 	
    urlbordercolor ={0 1 1} 	
}

\DeclareMathOperator{\Ad}{Ad}
\DeclareMathOperator{\ad}{ad}

\DeclareMathOperator{\Gr}{Gr}

\DeclareMathOperator{\haut}{ht}
\DeclareMathOperator{\id}{id}
\DeclareMathOperator{\im}{im}

\newcommand{\cA}{{\mathcal A}}

\newcommand{\cE}{{\mathcal E}}

\newcommand{\cG}{{\mathcal G}}
\newcommand{\cH}{{\mathcal H}}
\newcommand{\cJ}{{\mathcal J}}

\newcommand{\cL}{{\mathcal L}}
\newcommand{\cP}{{\mathcal P}}

\newcommand{\CC}{{\mathbb C}}

\newcommand{\NN}{{\mathbb N}}

\newcommand{\RR}{{\mathbb R}}

\newcommand{\ZZ}{{\mathbb Z}}

\newcommand{\fb}{{\mathfrak b}}
\newcommand{\fg}{{\mathfrak g}}
\newcommand{\fh}{{\mathfrak h}}

\newcommand{\fp}{{\mathfrak p}}

\newcommand{\set}[1]{\left\{ #1 \right\}}
\newcommand{\smooth}{C^{\infty}}

\theoremstyle{plain}
\newtheorem{thm}{Theorem}[subsection]
\newtheorem{lem}[thm]{Lemma}
\newtheorem{pro}[thm]{Proposition}
\newtheorem{cor}[thm]{Corollary}

\theoremstyle{definition}
\newtheorem{dfn}[thm]{Definition}

\theoremstyle{remark}

\newtheorem{rem}[thm]{Remark}

\numberwithin{equation}{subsection}
\numberwithin{figure}{subsection}

\begin{document}

%
%
\title{Courant Algebroids in Parabolic Geometry}

\author{Stuart Armstrong}
\author{Rongmin Lu}
\thanks{Corresponding author: Rongmin Lu. Email: \texttt{rongmin.lu@gmail.com}}

\begin{abstract}
Let $\fp$ be a Lie subalgebra of a semisimple Lie algebra $\fg$ and $(G,P)$ be the corresponding pair of connected Lie groups. A Cartan geometry of type $(G,P)$ associates to a smooth manifold $M$ a principal $P$-bundle and a Cartan connection, and a parabolic geometry is a Cartan geometry where $P$ is parabolic. We show that if $P$ is parabolic, the adjoint tractor bundle of a Cartan geometry, which is isomorphic to the Atiyah algebroid of the principal $P$-bundle, admits the structure of a (pre-)Courant algebroid, and we identify the topological obstruction to the bracket being a Courant bracket. For semisimple $G$, the Atiyah algebroid of the principal $P$-bundle associated to the Cartan geometry of $(G,P)$ admits a pre-Courant algebroid structure if and only if $P$ is parabolic.
\end{abstract}
\maketitle
\tableofcontents
%
%

\section{Introduction}

In this paper, we show that non-exact transitive Courant algebroids, possibly twisted by a closed 4-form, appear naturally in parabolic geometries. As a result, we are also able to show that the Atiyah algebroid of the principal $P$-bundle of a Cartan geometry of type $(G,P)$, for $G$ a semisimple Lie group, admits the structure of a pre-Courant algebroid if and only if $P$ is a parabolic subgroup of $G$. 

\subsection{Motivation}

Courant~\cite{Courant90} and Dorfman~\cite{Dorfman87} independently discovered the Courant bracket on the bundle $TM\oplus T^*M \to M$ while studying Dirac structures. Liu et al.~\cite{LiuWeinsteinXu97} later realised this bracket on the double $A\oplus A^*$ of a Lie bialgebroid $(A,A^*)$ and called the structure endowed by the bracket a \emph{Courant algebroid}. These prototypical Courant algebroids are called \emph{exact} because an exact Courant algebroid $E$ fits into an exact sequence of vector bundles
\[0\to A^* \to E \to A \to 0.\]
Exact Courant algebroids appeared most prominently in generalised complex geometry (Gualtieri~\cite{Gualtieri04,Gualtieri11}, Hitchin~\cite{Hitchin03}), where the study of the Courant algebroid introduced by Courant and Dorfman has led to interesting applications in mathematics and physics.

However, there are Courant algebroids which do not fit into such an exact sequence, since exactness is not a necessary consequence of the axioms defining a Courant algebroid. Indeed, exact Courant algebroids are the trivial examples of Courant algebroids, in the sense that Roytenberg~\cite{Roytenberg02} and \v{S}evera~\cite{Severa00} have shown that any exact Courant algebroid is a deformation of the standard Courant algebroid $TM\oplus T^*M$ by a closed 3-form.

Once we do not demand exactness, the next step is to consider \emph{transitive} Courant algebroids, where the anchor map $E\to TM$ is still surjective. Their structure has been clarified recently by Bressler~\cite{Bressler07} and Chen et al.~\cite{ChenStienonXu09}, following the outline sketched by \v{S}evera~\cite{Severa00}. It turns out that any transitive Courant algebroid $E$ is of the form 
\[E \cong TM\oplus \, \cE \, \oplus T^*M,\]
where $\cE$ is another vector bundle over $M$. Ginot and Gr\"{u}tzmann~\cite{GinotGrutzmann09} have also computed the cohomology of transitive Courant algebroids. Nevertheless, there appears to be very few examples of non-exact transitive Courant algebroids in the literature, outside of the work of Bursztyn et al.~\cite{BCG07} on reduction for Courant algebroids.

\subsection{Non-exact Courant algebroids in parabolic geometry}

The study of parabolic geometries is the study of Cartan connections on principal bundles that have a parabolic subgroup $P$ of a semisimple Lie group $G$ as their structure group. The notion of a parabolic geometry has emerged as a unifying framework and a powerful tool for the study of conformal geometry (Armstrong and Leistner~\cite{ArmstrongLeistner07}, Bailey et al.~\cite{BaileyEastwoodGover94}, \v{C}ap and Gover~\cite{CapGover02}), projective geometry (Eastwood~\cite{Eastwood08}), CR geometry (\v{C}ap~\cite{Cap02}, \v{C}ap and Gover~\cite{CapGover08}) and related contact structures (Fox~\cite{Fox05}) in differential geometry (see also the bibliography of \v{C}ap and Slov\'{a}k~\cite{CapSlovak09}). There is a rich representation theory associated to parabolic geometries, since the classification of parabolic subalgebras $\fp$ of semisimple Lie algebras $\fg$ (see \v{C}ap and Slov\'{a}k~\cite{CapSlovak09}) parallels the classification of semisimple Lie algebras.

In this paper, we show that transitive Courant algebroids that are \emph{a priori} non-exact, possibly with a bracket twisted by a closed 4-form, arise in parabolic geometries. Our work here extends the work done in an unpublished manuscript of Armstrong~\cite{Armstrong07}, the first author of this paper. Besides identifying the transitive Courant algebroid arising in a parabolic geometry, we have also isolated here the topological obstruction to the bracket being a Courant bracket, which is the closed 4-form alluded to earlier. This paper also takes a different approach towards deriving a formula for the bracket and our results hold without the technical assumptions, made by Armstrong~\cite{Armstrong07}, of regularity and normality on the parabolic geometry.

The candidate non-exact transitive Courant algebroid in a parabolic geometry is called the adjoint tractor bundle $\cA M$ (Definition \ref{dfn:adjoint_tractor}). Besides carrying the structure of a Courant algebroid, it is also naturally a Lie algebroid. In fact, $\cA M$ turns out to be the Atiyah algebroid $T\cP/P$ of the principal $P$-bundle $\cP$ defining the parabolic geometry. In addition, the Cartan connection of the parabolic geometry gives an isomorphism between the exact sequence of vector bundles
\[0\to \cP \times_P \fp \to \cA M \to \cP \times_P \fg/\fp \to 0,\]
induced from the exact sequence of vector spaces $0\to \fp \to \fg \to \fg/\fp \to 0$, and the Atiyah sequence (Atiyah~\cite{Atiyah57})
\[0 \to \cP \times_P \fp \to T\cP/P \to TM \to 0.\]
associated to $\cP$~\cite[Theorem 1]{Crampin09}. 

We also find in Corollary \ref{Courant_extension_AM} that we can realise $\cA M$ as a Courant extension (in the sense of Bressler~\cite{Bressler07}) 
\[0 \to T^*M \to \cA M \to T\cG_0/G_0 \to 0\]
of $T\cG_0/G_0$, which is the Atiyah algebroid given by the Atiyah sequence
\[0 \to \cG_0 \times_{G_0} \fg_0 \to T\cG_0/G_0 \to TM \to 0\]
of the principal bundle $\cG_0$ with structure group $G_0\leq P$. Following the methods of \v{S}evera~\cite{Severa00}, Bressler~\cite{Bressler07} and Chen et al.~\cite{ChenStienonXu09}, this leads us to derive detailed formulas for the Courant bracket on $\cA M$ (Theorem \ref{main_theorem}) from the Courant algebroid axioms. 

As alluded to in the above, however, there is an obstruction to the Courant bracket satisfying the Jacobi identity, which is the first Pontrjagin class (more precisely, half the first Pontrjagin class) associated to the Atiyah algebroid of $\cG_0$. When this does not hold, the bracket satisfies a twisted Jacobi identity and we get a pre-Courant bracket on a pre-Courant algebroid, in the sense of Vaisman~\cite{Vaisman05}.

\subsection{Outline of the paper}

We begin with a brief review of the theory of Lie algebroids in \S \ref{sec:Atiyah}, with an emphasis on the Atiyah sequence. We then consider Cartan and parabolic geometries, and the associated adjoint tractor bundles, in \S \ref{sec:adjoint_tractor}. Our treatment is rather brief and the reader should refer to the seminal work of \v{C}ap and Slov\'{a}k~\cite{CapSlovak09} or the papers referenced therein (e.g. \v{C}ap and Gover~\cite{CapGover02}) for further details. The Weyl structure of a parabolic geometry is then introduced to identify the adjoint tractor bundle $\cA M$ with its associated graded bundle $\Gr(\cA M)$. 

We then consider the theory of (pre-)Courant algebroids in \S \ref{sec:structure}. It turns out that $\cA M$ has the same underlying vector bundle structure as a transitive Courant algebroid, so we are able to derive explicit formulas for the pre-Courant bracket on $\cA M$. The Weyl structure of a parabolic geometry also allows us to relate the 3-form $\cH$ in the bracket to the first Pontrjagin class associated to the Atiyah algebroid of the principal $G_0$-bundle $\cG_0$. This gives a topological condition for the bracket to satisfy the Jacobi identity and an interesting characterisation of parabolic geometries with flat Weyl connections. We then prove Theorem \ref{classification}, which shows, in particular, that parabolic geometries supply examples of pre-Courant algebroid structures on Atiyah algebroids that seem to have been missing so far from the literature.

Finally, we discuss in \S \ref{sec:Conclusion} some of the connections the present work has with recent developments in related areas and conclude with speculations for future work.

Throughout this paper, manifolds are assumed to be smooth, finite-dimensional, orientable, second countable and paracompact.

\section{The Atiyah sequence and principal connections}\label{sec:Atiyah}

We review the construction of the Atiyah sequence of a principal bundle $\cP$ over a smooth manifold $M$, recast the definition of a principal connection in terms of the Atiyah sequence, and consider the existence of isotropic principal connections with respect to a non-degenerate metric. Most of the material in this section is standard, so our treatment here shall be brief.

\subsection{The Atiyah sequence}

\begin{dfn}\label{dfn:Lie_algebroid}
A \emph{Lie algebroid} on $M$ is a vector bundle $E\to M$ equipped with a bundle map $\alpha\colon E \to TM$, called the \emph{anchor} of $E$, and a skew-symmetric $\RR$-bilinear bracket on the sections of $E$, $[-,-]\colon \Gamma E \times \Gamma E \to \Gamma E$, which satisfies the following:
\begin{enumerate}
	\item The Jacobi identity, i.e. for all $X,Y,Z\in\Gamma E$,
	\[[X,[Y,Z]] = [[X,Y],Z] + [Y,[X,Z]].\]
	
	\item The Leibniz rule, i.e. for all $X,Y\in\Gamma E$, $f\in \smooth(M)$,
	\[[X,fY] = f[X,Y] + \alpha(X)(f) Y.\]
	
	\item The bracket commutes with the anchor, i.e. for all $X,Y\in\Gamma E$,
	\[\alpha([X,Y]) = [\alpha(X),\alpha(Y)].\]
\end{enumerate}
\end{dfn}

Some basic examples are Lie algebras and the tangent bundle of $M$. The example that is most relevant for our purposes, however, is the Atiyah sequence associated to a principal $H$-bundle $\cP$ over $M$. 

Consider a principal $H$-bundle $\pi\colon\cP \to M$, with the Lie group $H$ acting on $\cP$ from the right. The right action of $H$ on $\cP$ induces a right action of $H$ on $T\cP$. Sections of the bundle $T\cP/H$ over $\cP/H=M$ are just $H$-invariant vector fields on $P$, so this is a vector bundle over $M$. The kernel of the projection $\alpha:=\pi_* \colon T\cP/H \to TM$ is, by definition, the bundle of vertical $H$-invariant vector fields on $\cP$, $V\cP/H$. Thus, we obtain the exact sequence of vector bundles
\[0 \to V\cP/H \hookrightarrow T\cP/H \xrightarrow{\alpha} TM \to 0.\]

This could be called the Atiyah sequence, but it is usual and more useful to identify $V\cP/H$ with $\cP \times_H \fh$, which is the vector bundle associated to $\cP$ by the adjoint representation of $H$ on $\fh$. First, we recall the following definition.

\begin{dfn}\cite[\S3.9]{GreubHalperinVanstone73} \label{dfn:fundamental}
Let $R\colon \cP \times H \to \cP$ be the right action on $\cP$ and $T$ be the tangent map. The map defined by
\[j\colon \cP \times \fh \to T\cP, \quad (u,X) \mapsto T(R)_{(u,e)}(0,X),\]
where $e$ is the identity of $H$, is called the fundamental vector field map. The image $j(X)$ of $X\in\fh$ is referred to as the \emph{fundamental vector field} generated by $X$.
\end{dfn}

We state the following result (\cite[Proposition A.3.2]{Mackenzie87} or \cite[Proposition 3.2.2]{Mackenzie05}).

\begin{pro}
The fundamental vector field map is a vector bundle isomorphism over $M$ identifying $\cP \times_H \fh$ with $V\cP/H$.
\end{pro}

Hence, we have the desired exact sequence:

\begin{dfn}\label{dfn:Atiyah_sequence}
Let $\pi\colon \cP\to M$ be a principal $H$-bundle and $\fh$ be the Lie algebra of $H$. The exact sequence of vector bundles over $M$
\[0 \to \cP \times_H \fh \xrightarrow{j} T\cP/H \xrightarrow{\alpha} TM \to 0\]
is the \emph{Atiyah sequence} of $\cP$. We shall call the bundle $T\cP/H$ the \emph{Atiyah algebroid} of $\cP$. The bundle $\cP\times_H \fh$ is usually referred to as the \emph{adjoint bundle} of $\cP$.
\end{dfn}

\begin{rem}
In fact, the Atiyah algebroid of $\cP$ is a Lie algebroid, as sections of $T\cP/H$ can be identified with $H$-invariant vector fields on $\cP$, which are closed under the Lie bracket and are projectable onto $M$. For further details, the reader is referred to \cite[pp. 93--97]{Mackenzie05}.
\end{rem}

\subsection{Principal and Lie algebroid connections}

We now turn to the relationship between Lie algebroid connections of the Atiyah algebroid and principal connections on $\cP$. 

\begin{dfn}\label{dfn:principal_connection}
A \emph{principal connection} on $\cP$ is a smooth horizontal distribution $\cH\cP\subset T\cP$, such that $T\cP \cong \cH\cP \oplus V\cP$ and $\cH\cP$ is invariant under the right action of the structure group $H$. A \emph{principal connection (form)} on $\cP$ is a 1-form $\omega\in\Omega^1(\cP,\fh)$ that satisfies the following conditions:
\begin{enumerate}
	\item It reproduces the generators of fundamental vector fields, i.e. for any $X\in \fh$, \[\omega(j(X)) = X,\] where $j(X)$ is the fundamental vector field (Definition \ref{dfn:fundamental}) generated by $X$.\\

	\item It is $H$-equivariant, i.e. for any $h\in H$, \[(R_h)^*\omega = \Ad(h^{-1})\omega,\] where $R_h$ is the right action of $H$ on $\cP$ and $\Ad$ is the adjoint action of $H$ on $\fh$.
\end{enumerate}
\end{dfn}

Next, we recall the notion of a Lie algebroid connection of the Atiyah algebroid.

\begin{dfn}\label{dfn:Lie_algebroid_connection}
A \emph{(Lie algebroid) connection} of the Atiyah algebroid of $\cP$ is a section $\gamma\colon TM\to T\cP/H$ of the anchor $\alpha\colon T\cP/H \to TM$. A \emph{connection reform} is a vector bundle morphism $\omega\colon T\cP/H \to \cP\times_H \fh$ such that $\omega \circ j = \id_{\cP\times_H \fh}$, where $j\colon \cP\times_H \fh \to T\cP/H$ is the fundamental vector field map.
\end{dfn}

The following result on Lie algebroid connections, which is a specialisation of \cite[Proposition 5.2.6]{Mackenzie05}, is well-known.

\begin{pro}\cite[Proposition 5.2.6]{Mackenzie05}
For the Atiyah sequence of a principal bundle $\cP\to M$,
\[0 \to \cP \times_H \fh \xrightarrow{j} T\cP/H \xrightarrow{\alpha} TM \to 0,\]
there exist a right inverse $\gamma\colon TM\to T\cP/H$ to $\alpha$ and a left inverse $\omega\colon T\cP/H \to \cP\times_H \fh$ to $j$. In addition, if either a right inverse $\gamma$ or a left inverse $\omega$ is given, the other can be uniquely chosen to satisfy the condition
\[j\circ \omega + \gamma\circ \alpha = \id_{T\cP/H}.\]
Thus, there is a one-to-one correspondence between connections and connection reforms of the Atiyah algebroid of $\cP$.
\end{pro}

It follows that principal connections on $\cP$ are equivalent to connections of the Atiyah algebroid of $\cP$.

\begin{pro}\cite[Proposition 5.3.2]{Mackenzie05} \label{identify_Lie_algebroid_connections}
There is a one-to-one correspondence between connections of the Atiyah algebroid of $\cP$ and principal connections on $\cP$.
\end{pro}

In the following, we shall assume that the principal $H$-bundle $\cP$ is equipped with a non-degenerate metric $(-,-)$, which is invariant under the right action of $H$.

\begin{rem}
The metric induces the usual \emph{musical isomorphisms}. For any $u_1,u_2\in T\cP$, $\flat\colon T\cP \to T^*\cP$ is given by
\[\langle\flat(u_1),u_2\rangle = (u_1,u_2),\]
and its inverse is $\sharp\colon T^*\cP \to T\cP$.
\end{rem}

\begin{dfn}
Suppose there is a short exact sequence of vector bundles
\[0 \to E' \to E \xrightarrow{p} E'' \to 0\]
over a smooth manifold $M$, and that $E$ is equipped with a non-degenerate metric $(-,-)$. An \emph{isotropic splitting} of this sequence is a section $\gamma\colon E'' \to E$ of $p\colon E\to E''$ such that its image $\gamma(E'')$ is isotropic in $E$, i.e. the metric vanishes on $\gamma(E'')$, or $\gamma(E'') \subseteq \gamma(E'')^{\perp}$.
\end{dfn}

\begin{pro}\label{existence_isotropic_splitting}
An isotropic splitting of the Atiyah sequence exists.
\end{pro}

\begin{proof}
Consider any section $\gamma\colon TM \to T\cP/H$ of the anchor $\alpha\colon T\cP/H \to TM$. For $x,y\in \Gamma(TM)$, we define the bundle map $\eta\colon TM \to T^*M$ by
\[\langle \eta(x), y \rangle = \frac{1}{2}(\gamma(x),\gamma(y)),\]
where $\langle -,- \rangle$ is the natural pairing. Then $\gamma_0 = \gamma - \sharp\circ \alpha^*\circ \eta$ is an isotropic splitting.
\end{proof}

\begin{cor}
Isotropic principal connections exist.
\end{cor}

\begin{proof}
This is a consequence of Propositions \ref{identify_Lie_algebroid_connections} and \ref{existence_isotropic_splitting}.
\end{proof}

\section{The adjoint tractor bundle}\label{sec:adjoint_tractor}

In this section, we study the adjoint tractor bundle of a parabolic geometry over a smooth manifold $M$. First, we introduce Cartan geometries in \S \ref{subsec:Cartan} and show that $\cA M$ can be identified with the Atiyah algebroid of the principal bundle $\cP$ of the Cartan geometry. Next, in \S \ref{subsec:parabolic}, we consider the special case of a parabolic geometry, which endows the adjoint tractor bundle $\cA M$ with a rich structure. In particular, we show in \S \ref{subsec:Weyl_structures} that there is an associated graded bundle $\Gr(\cA M)$, which can be identified with the ordinary adjoint tractor bundle using a Weyl structure.

\subsection{Cartan geometry}\label{subsec:Cartan}

We begin by considering a closed subgroup $H$ of a semisimple Lie group $G$, with the corresponding Lie algebras $\fh \subset \fg$ satisfying the condition that $\dim \fg/\fh = \dim M$.

\begin{dfn}\label{dfn:Cartan_connection}
Suppose $\pi\colon\cP\to M$ is a principal $H$-bundle over $M$. With the same notation as in Definition \ref{dfn:principal_connection}, we say that a 1-form $\omega \in \Omega^1(\cP, \fg)$ is a \emph{Cartan connection} if $\omega$ satisfies the following three conditions:

\begin{enumerate}
	\item (\emph{$H$-equivariance}) For each $h\in H$, \[(R_h)^*\omega = \Ad(h^{-1})\omega.\]
	\item (\emph{Reproduces fundamental vector fields}) For each $X\in\fh$ and $j(X)$ the corresponding fundamental vector field, \[\omega(j(X)) = X.\]
	\item (\emph{Absolute parallelism}) For each $u\in \cP$, the restriction \[\omega|_u\colon T_u\cP \to \fg\] is a linear isomorphism.
\end{enumerate}
\end{dfn}

\begin{dfn}\label{dfn:Cartan_geometry}
A \emph{Cartan geometry} of type $(G,H)$ is a pair $(\pi\colon\cP\to M, \omega)$, where $\cP$ is a principal $H$-bundle over $M$ and $\omega\in \Omega^1(\cP,\fg)$ is a Cartan connection. We shall write $\cP$ for the Cartan geometry over $M$ where no confusion could arise in the following.
\end{dfn}

The following Lemma exhibits an important associated bundle of a Cartan geometry.

\begin{lem}\cite{CapSlovak09,Sharpe97}.\label{identify_TM}
Let $\cP$ be a Cartan geometry of type ($G,H$). Then there is an isomorphism $\cP \times_H \fg/\fh \cong TM$, where $H$ acts by the restriction of the adjoint action of $G$.
\end{lem}

\begin{proof}
Consider the map
\begin{align*}
\phi \colon \cP \times \fg &\to TM\\
(u, X) & \mapsto (\pi(u), \pi_* \omega_u^{-1}(X)).
\end{align*}
By Definition \ref{dfn:Cartan_connection}(2), if $X \in \fh$, then $\omega^{-1}(X)$ is the fundamental vector field generated by $X$, which is vertical. Hence, $\phi$ factors to a map $\cP \times \fg/\fh \to TM$. The $H$-equivariance of $\omega$ means $\phi$ factors further to the vector bundle map $\cP\times_H \fg/\fh \to TM$ over $M$, which is an isomorphism on each fibre and hence identifies the two bundles.
\end{proof}

We can now introduce the adjoint tractor bundle, the main object of our study.

\begin{dfn}\label{dfn:adjoint_tractor}
Let $(\cP \to M, \omega)$ be a Cartan geometry of type ($G,H$). The \emph{adjoint tractor bundle} is the vector bundle $\cA M:= \cP\times_H \fg$ associated to the adjoint representation of $H$ on $\fg$.
\end{dfn}

\begin{rem}
We need to distinguish between the adjoint bundle, which is the vector bundle associated to the adjoint representation of $H$ on $\fh \subset \fg$ in the Atiyah sequence, and the adjoint \emph{tractor} bundle. 
\end{rem}

The following Proposition is a reformulation of \cite[Theorem 1]{Crampin09}, which identifies the adjoint tractor bundle with the Atiyah algebroid of the principal bundle associated to the Cartan geometry.

\begin{pro}
Let  $(\pi\colon\cP\to M, \omega)$ be a Cartan geometry of type $(G,H)$. Then the Cartan connection $\omega$ induces an isomorphism of exact sequences of vector bundles:\\

\begin{center}
$\begin{CD}
0 @>>> \cP \times_H \fh @>>> \cP \times_H \fg @>>> \cP \times_H \fg/\fh @>>> 0\\
@. @| @VVV @VVV @.\\
0 @>>> \cP \times_H \fh @>>> T\cP/H @>\pi_*>> TM @>>> 0\\
@. @. @. @. @.
\end{CD}$
\end{center}

In particular, the adjoint tractor bundle $\cA M = \cP \times_H \fg$ can be identified with the Atiyah algebroid, contains $TM$ as a quotient bundle, and has the adjoint bundle as a sub-bundle.

\end{pro}

\begin{proof}
There is an isomorphism $\omega^{-1}\colon \cP\times \fg \to T\cP$, which is $H$-equivariant and so factors to the required isomorphism of vector bundles $\cP\times_H \fg \to T\cP/H$. We have already seen in Lemma \ref{identify_TM} that there is an isomorphism $\phi\colon \cP\times_H \fg/\fh \to TM$, which completes the proof.
\end{proof}

\subsection{Parabolic geometry}\label{subsec:parabolic}

A richer structure on the adjoint tractor bundle can be obtained by considering parabolic geometries. We begin by considering parabolic Lie subalgebras $\fp$ of a semisimple Lie algebra $\fg$. To simplify our exposition, we shall consider only complex semisimple Lie algebras (but see Remark \ref{rem:real_complex}).

Recall that for a complex semisimple Lie algebra $\fg$, we can choose a Cartan subalgebra $\fh\leq \fg$ and consider the root spaces of the adjoint representation
\[\fg_{\eta} = \set{X\in \fg\,|\: (\ad(Y)-\eta(Y)I)X=0, \quad \forall Y\in\fh}\]
where $\eta\in\fh^*$. This gives us a root space decomposition of $\fg$,
\[\fg = \fh \oplus \bigoplus_{\eta\in\Delta}\fg_{\eta},\]
where $\fh = \fg_0$ and $\Delta = \set{0\neq \eta \in\fh^*\,|\: g_{\eta} \neq \emptyset}$ is the root system of $\fg$.

\begin{rem}
It is easy to see that $[\fg_{\eta},\fg_{\theta}] \subset \fg_{\eta+\theta}$, provided $\eta+\theta \in \Delta$. It follows that, for any $X\in\fg_{\eta}$ and $Y\in\fg_{\theta}$ with $\eta+\theta \in \Delta$, $\ad(X)\ad(Y)$ is nilpotent unless $\eta+\theta = 0$, since $\Delta$ is a finite set. Hence, the Killing form is degenerate on $\fg_{\eta} \times \fg_{\theta}$ if $\eta+\theta \neq 0$.
\end{rem}

We can choose a basis of simple roots $\Delta^0$ for the root system of $\fg$, which gives us a partial ordering on the roots: for $\alpha,\beta\in\Delta$, $\alpha\preceq \beta$ if and only if $\beta-\alpha$ can be written in terms of the simple roots with non-negative coefficients. The set of roots
\[\Delta^+ = \set{\eta\in\Delta\,|\: \eta \succ 0}\]
is then called the system of positive roots. This is encapsulated by the choice of a Borel subalgebra $\fb$ of $\fg$, 
\[\fb = \fh \oplus \bigoplus_{\eta\in\Delta^+} \fg_{\eta}.\]
Note that $\fb$ is a maximal solvable subalgebra of $\fg$. The Borel subalgebra $\fb$ associated in this way to $\fh$ and $\Delta^+$ is called a standard Borel subalgebra. 

\begin{dfn}
Let $\fg$ be a complex semisimple Lie algebra. A \emph{parabolic subalgebra} $\fp$ of $\fg$ is a Lie subalgebra containing a Borel subalgebra.
\end{dfn}

Note that, just as any Borel subalgebra is conjugate to a standard one, any parabolic subalgebra is conjugate to a standard one as well.

We can now state the following result.

\begin{pro}\cite[Proposition 3.2.1]{CapSlovak09}
Let $\fg$ be a complex semisimple Lie algebra, $\fh$ be a Cartan subalgebra of $\fg$ and $\Delta^0$ be a choice of a set of simple roots. Then standard parabolic subalgebras $\fp\leq \fg$ are in one-to-one correspondence with subsets $\Sigma$ of $\Delta^0$. More precisely, to each $\fp$, we can associate the subset $\Sigma_{\fp} = \set{\eta\in\Delta^0\,|\: \fg_{-\eta}\nsubseteq \fp}$.
\end{pro}

It turns out that choosing a parabolic subalgebra $\fp$ of a semisimple Lie algebra $\fg$ leads to a grading on $\fg$. We give a definition of this grading. 

\begin{dfn}\label{dfn:k-grading}
Suppose $k\in \NN$ is nonzero and $\fg$ is a semisimple Lie algebra. A \emph{$|k|$-grading} on $\fg$, which makes $\fg$ into a \emph{$|k|$-graded algebra}, is a decomposition of $\fg$ into a direct sum of subspaces 
\[\fg \cong \fg_{-k} \oplus \cdots \oplus \fg_{-1} \oplus \fg_0 \oplus \fg_1 \oplus \cdots \oplus \fg_k\]
that satisfies the following conditions:
\begin{enumerate}
	\item For any $i,j\in \ZZ$, $[\fg_i,\fg_j] \subset \fg_{i+j}$, and whenever $|i|>k$, $\fg_i = \{0\}$.
	\item As a Lie algebra, the subalgebra $\fg_- := \fg_{-k} \oplus \cdots \oplus \fg_{-1}$ is generated by $\fg_{-1}$.
	\item Neither $\fg_{-k}$ nor $\fg_k$ is equal to $\{0\}$.
\end{enumerate}
The associated filtration is
\[\fg_k= \fg^k \subset \fg^{k-1} \subset \cdots \subset \fg^{-k} = \fg, \quad \fg^i := \bigoplus_{j\geq i} \fg_j.\]
We shall write $\fp := \fg^0 = \fg_0 \oplus \fg_1 \oplus \cdots \oplus \fg_k$ and $\fp_+ := \fg^1 = \fg_1 \oplus \cdots \oplus \fg_k$ for the nilpotent ideal in $\fp$.
\end{dfn}

There is a natural grading that can be constructed from the data for specifying a parabolic subalgebra.

\begin{dfn}\label{dfn:Sigma_height}
Given a subset $\Sigma\subset \Delta^0 = \set{\eta_1,\ldots,\eta_m}$ of simple roots corresponding to a parabolic subalgebra $\fp\leq \fg$, the \emph{$\Sigma$-height} $\haut_{\Sigma}(\eta)$ of a root 
\[\eta = \sum_{i=1}^m b_i\eta_i, \quad \eta\in\Delta,\]
is given by
\[\haut_{\Sigma}(\eta) := \sum_{\eta_i\in\Sigma} b_i.\]
The corresponding root space $\fg_{\eta}$ is also said to have a $\Sigma$-height of $\haut_{\Sigma}(\eta)$.
\end{dfn}

For each integer $i$, we can then define the space $\fg_i$ to be the direct sum of all root spaces of $\Sigma$-height $i$, except for $\fg_0$, which is given by the direct sum of the Cartan subalgebra and all root spaces of $\Sigma$-height $0$. The ordering on the roots then gives us a maximal $\Sigma$-height $k$. We can now state the following Proposition (see \v{C}ap and Slov\'{a}k~\cite{CapSlovak09} for a proof).

\begin{pro}\cite[Theorem 3.2.1(1)]{CapSlovak09}\label{parabolic_k-grading}
Let $\fg$ be a complex semisimple Lie algebra, $\fh$ be a Cartan subalgebra of $\fg$ and $\Delta^0$ be a choice of a set of simple roots. Given any standard parabolic subalgebra $\fp\leq \fg$ associated to the subset $\Sigma\subseteq \Delta^0$, decomposing $\fg$ according to $\Sigma$-height
\[\fg = \fg_{-k} \oplus \cdots \oplus \fg_k\]
equips $\fg$ with a $|k|$-grading such that $\fp=\fg^0$.
\end{pro}

\begin{rem}\label{rem:real_complex}
We have confined ourselves to treating the case of complex semisimple Lie algebras, but a similar version of Proposition \ref{parabolic_k-grading} holds for the real case. In fact, the result in Proposition \ref{parabolic_k-grading} can be strengthened: the choice of a parabolic subalgebra $\fp$ of a real or complex semisimple Lie algebra $\fg$ is equivalent to the choice of a $|k|$-grading for $\fg$, modulo some technical conditions. For our purposes, however, it suffices to know that parabolic subalgebras of $\fg$ give us $|k|$-gradings of $\fg$. We refer the reader to \cite[\S 3.2]{CapSlovak09} for a more comprehensive treatment.
\end{rem}

We collect some useful properties of a $|k|$-graded algebra.

\begin{pro} \cite[Proposition 3.1.2]{CapSlovak09} \label{k-grading_properties}
Suppose $\fg$ is a $|k|$-graded semisimple Lie algebra 
\[\fg \cong \fg_{-k} \oplus \cdots \oplus \fg_{-1} \oplus \fg_0 \oplus \fg_1 \oplus \cdots \oplus \fg_k\]
and $B$ is the Killing form on $\fg$. Then the following statements are true:
\begin{enumerate}
	\item There is a unique \emph{grading element} $K\in\fg$, which lies in the centre of $\fg_0$ and satisfies the identity $[K,X] = jX$ for all $X\in\fg_j, j=-k,\ldots,k$.
	\item The duality $B\colon \fg \xrightarrow{\cong} \fg^*$ is compatible with the filtration and the grading of $\fg$. Hence, for each $j=-k,\ldots,k$, $B$ gives a duality of $\fg_0$-modules $B\colon \fg_j \xrightarrow{\cong} \fg_{-j}$ and a duality of $\fp$-modules $B\colon \fg/\fg^{-j+1} \xrightarrow{\cong} \fg^j$. In particular, $\fg/\fp \cong \fp_+$.
\end{enumerate}
\end{pro}

\begin{proof}
We sketch a proof. For part 1, note that the map $E\colon \fg \to \fg$ given by
\[E(X) = jX, \qquad (X\in \fg_j, \: j=-k,\ldots,k)\]
is a derivation by Definition \ref{dfn:k-grading}(1). As $\fg$ is semisimple, there exists a unique element $K\in \fg$ such that $E(X) = [K,X]$. We then observe that the decomposition of $K=K_{-k} + \ldots + K_k$ by grading shows that
\[[K,K] = \sum_{j=-k}^{k} jK_j = 0,\]
so K must lie in $\fg_0$. Note also that for all $X\in\fg_0$, $[K,X]=0$, so $K$ must lie in the centre of $\fg_0$.

For part 2, first recall that the invariance of $B$ as a bilinear form means that, for any $X,Y\in \fg$,
\[B([K,X],Y) = -B(X,[K,Y]).\]
If $X\in \fg_i$ and $Y\in\fg_j$, then $(i+j)B(X,Y) = 0$, so if $i+j\neq 0$, then $B(X,Y)$ must vanish. This implies that only the restrictions of $B$ to $\fg_j \times \fg_{-j}$  $(j=0,1,\ldots,k)$ are non-degenerate, so the duality induced by $B$ is compatible with the grading. In particular $\fg_j$ is dual to $\fg_{-j}$.

To see the compatibility with the filtration, observe that $\fg/\fg^{-i+1} \cong \fg_{-k} \oplus \cdots \oplus \fg_{-i}$ as vector spaces, so it has the same dimension as $\fg^i=\fg_i\oplus \cdots \oplus \fg_k$. Moreover, $B$ is clearly non-degenerate on $\fg/\fg^{-i+1} \times \fg^i$, so this induces a duality between these two spaces. In particular, $\fg/\fg^0 = \fg/\fp$ is dual to $\fg^1=\fp_+$, and this completes the proof of (2).
\end{proof}

We can now describe the Lie groups corresponding to the Lie algebras $\fg_0 \subset \fp$.

\begin{lem}\cite[Lemma 3.1.3]{CapSlovak09} \label{parabolic_Levi_group}
Let $\Ad$ be the adjoint representation of a semisimple Lie group $G$ on its Lie algebra $\fg$, and $\ad$ be the adjoint representation of $\fg$. Suppose there is a $|k|$-grading on $\fg$ such that we have the Lie subalgebras $\fg_0\subset \fp \subset \fg$. Then we have
\[\fg_0 = \set{X\in\fg\,|\: ad(X)(g_i)\subset \fg_i \quad \forall i=-k,\ldots,k},\]
\[\fp = \set{X\in\fg\,|\: ad(X)(g^i)\subset \fg^i \quad \forall i=-k,\ldots,k},\]
and the corresponding Lie groups are
\[G_0 = \set{g\in G\,|\: \Ad(g)(\fg_i) \subset \fg_i \quad \forall i=-k,\ldots,k}\]
and
\[P = \set{g\in G\,|\: \Ad(g)(\fg^i) \subset \fg^i \quad \forall i=-k,\ldots,k}.\]
Thus, $P$ is a \emph{parabolic subgroup} of $G$ corresponding to the $|k|$-grading of $\fg$.
\end{lem}

\begin{dfn}\label{dfn:parabolic_geometry}
Let $P$ be a parabolic subgroup of $G$, $\fp\subset \fg$ be the corresponding Lie algebras and $\cP$ be a principal $P$-bundle over a smooth manifold $M$, with $\dim \fg/\fp = \dim M$. A \emph{parabolic geometry} on $M$ is a Cartan geometry $(\pi\colon\cP\to M, \omega)$ of type $(G,P)$.
\end{dfn}

Recall that the adjoint tractor bundle for a parabolic geometry is then defined to be
\[\cA M := \cP \times_P \fg.\]
We note that Lemma \ref{parabolic_Levi_group} implies that the filtration is only $P$-invariant, while the grading is $G_0$-invariant. Nevertheless, we can see that the $P$-invariant filtration
\[\fg= \fg^{-k} \supset \cdots \supset \fg^k\]
induces a corresponding filtration
\[\cA M = \cA^{-k} M \supset \cdots \supset \cA^k M\]
in the adjoint tractor bundle. In particular, we note that
\[\cA^0 M = \cP \times_P \fp, \quad \cA^1 M = \cP \times_P \fp_+,\]
so the quotient bundle $\cA M / \cA^0 M$ is just
\[\cA M / \cA^0 M = \cP \times_P \fg/\fp \cong TM,\]
and by Proposition \ref{k-grading_properties}(2), $\cA^1 M \cong T^*M$. The associated graded bundle $\Gr(\cA M)$ has components $\Gr_i(\cA M) = \cA^i M / \cA^{i+1} M = \cP \times_P (\fg^i/\fg^{i+1})$, so
\[\Gr(\cA M) \cong TM \oplus (\cP\times_P \fg_0) \oplus T^*M.\]
However, as the grading is $G_0$-invariant, this should be defined in terms of a $G_0$-bundle. We shall set out to do this in the following.

\subsection{Weyl structures in parabolic geometry}\label{subsec:Weyl_structures}

We recall the notion of a Weyl structure of a parabolic geometry. This reduces the principal $P$-bundle $\cP$ of the parabolic geometry to a principal $G_0$-bundle $\cG_0$, which allows us to identify the adjoint tractor bundle $\cA M$ with its associated graded bundle $\Gr(\cA M)$. It also allows us to show that the adjoint tractor bundle is actually an extension of the Atiyah algebroid of the principal $G_0$-bundle $\cG_0$.

We begin by considering the underlying principal $G_0$-bundle $\pi_0\colon \cG_0 \to M$. This can be obtained from the principal $P$-bundle by forming the quotient $\cG_0:= \cP/P_+$, where $P_+$ is the nilpotent normal subgroup of $P$ with Lie algebra $\fp_+$, and yields a projection $p\colon \cP \to \cG_0$.

\begin{dfn}
Let $(\pi\colon \cP\to M, \omega)$ be a parabolic geometry. A local \emph{Weyl structure} for $\cP$ is a local smooth $G_0$-equivariant section $\sigma\colon \cG_0 \to \cP$ of $p\colon \cP\to \cG_0$.
\end{dfn}

What we require, however, is a global Weyl structure. In the smooth category, this always exists, and we state the following result, which can be proved using the usual partition of unity argument.

\begin{lem}\cite[Proposition 5.1.1]{CapSlovak09}
There exists a global Weyl structure $\sigma\colon \cG_0\to \cP$ for any smooth parabolic geometry $(\pi\colon \cP\to M, \omega)$.
\end{lem}

This enables us to identify $\cA M$ with its associated graded $\Gr(\cA M)$.

\begin{cor}\cite[Corollary 5.1.3]{CapSlovak09} \label{AM_with_Gr}
Let $(\pi\colon \cP\to M, \omega)$ be a parabolic geometry and $\cA M=\cP\times_P \fg$ be the bundle associated to the adjoint representation of $P$ on $\fg$. Then $\Gr(\cA M)$ can be naturally identified with $\cG_0 \times_{G_0} \fg$ and a Weyl structure gives an isomorphism between $\cA M$ and $\Gr(\cA M)$,
\[\cA M \cong \Gr(\cA M) \cong TM \oplus \cG_0 \times_{G_0} \fg_0 \oplus T^*M,\]
where $TM \cong \cG_0 \times_{G_0} \fg/\fp$ and $T^*M \cong \cG_0 \times_{G_0} \fp_+$.
\end{cor}

We also have the following useful result on principal connections.

\begin{pro}\cite[Proposition 5.1.2]{CapSlovak09} \label{Weyl_connection}
Suppose we are given a Weyl structure $\sigma\colon \cG_0\to \cP$ on a parabolic geometry $(\pi\colon \cP\to M, \omega)$. Then the pullback of the Cartan connection, $\sigma^*\omega\in\Omega^1(\cG_0,\fg)$, decomposes as
\[\sigma^*\omega = \sigma^*\omega_{-k} + \cdots + \sigma^*\omega_0 + \cdots + \sigma^* \omega_k\]
and $\sigma^*\omega_0\in\Omega^1(\cG_0,\fg_0)$ gives a principal connection on the principal $G_0$-bundle $\pi_0\colon \cG_0\to M$.
\end{pro}

\begin{proof}
We sketch a proof here, referring the reader to \v{C}ap and Slov\'{a}k~\cite{CapSlovak09} for details. The $G_0$-equivariance of $\sigma$ induces $G_0$-equivariance on $\sigma^*\omega$. The $|k|$-grading on $\fg$ induces the stated decomposition and each of the components is $G_0$-equivariant, so $\sigma^*\omega_0$ is $G_0$-equivariant too. Finally, since the Cartan connection reproduces generators of fundamental vector fields and $\sigma$ is $G_0$-equivariant, $\sigma^*\omega_0$ also reproduces generators of fundamental vector fields. Hence, $\sigma^*\omega_0$ is a principal connection on $\cG_0$.
\end{proof}

\begin{cor}\label{Atiyah_of_G0}
The principal connection $\sigma^*\omega_0$ on $\cG_0$ is a splitting of the Atiyah sequence of $\cG_0$, which is given by
\[0 \to \cG_0\times_{G_0} \fg_0 \to \cG_0\times_{G_0} \fp \to TM \to 0.\]
\end{cor}

\begin{proof}
By Proposition \ref{identify_Lie_algebroid_connections}, we know that $\sigma^*\omega_0$ is also a connection of the Atiyah algebroid, and hence splits the Atiyah sequence of $\cG_0$, 
\[0 \to \cG_0 \times_{G_0} \fg_0 \to T\cG_0/G_0 \to TM \to 0.\]
We have already identified $TM\cong \cG_0\times_{G_0} \fg/\fp$, but since $\fp_+ \cong \fg/\fp$ (Proposition \ref{k-grading_properties}(2)) and $\fp = \fg_0 \oplus \fp_+$, we can rewrite the Atiyah sequence as
\[0 \to \cG_0\times_{G_0} \fg_0 \to \cG_0\times_{G_0} \fp \to TM \to 0,\]
which completes the proof.
\end{proof}

\begin{dfn}\label{dfn:Weyl_connection}
The principal connection $\sigma^*\omega_0$ on $\cG_0$ is called the \emph{Weyl connection} associated to $\sigma$.
\end{dfn}

We now recall the following definition from Bressler~\cite{Bressler07}.

\begin{dfn}\label{dfn:Courant_extension}
A \emph{Courant extension} of a Lie algebroid $\cA$ is a Courant algebroid $E$ with an isomorphism of Lie algebroids $E/T^*M \cong \cA$.
\end{dfn}

The next statement follows immediately from Corollaries \ref{AM_with_Gr} and \ref{Atiyah_of_G0}.

\begin{cor}\label{Courant_extension_AM}
The adjoint tractor bundle $\cA M$ is a Courant extension
\[0 \to T^*M \to \cA M \to T\cG_0/G_0 \to 0\]
of the Atiyah algebroid $T\cG_0/G_0$ associated to the Atiyah sequence
\[0 \to \cG_0 \times_{G_0} \fg_0 \to T\cG_0/G_0 \to TM \to 0\]
of the principal $G_0$-bundle $\cG_0$ over $M$.
\end{cor}

\section{Transitive Courant algebroids}\label{sec:structure}

After a review of the theory of transitive (pre-)Courant algebroids, following Kosmann-Schwarzbach~\cite{Kosmann05}, Uchino~\cite{Uchino02} and Vaisman~\cite{Vaisman05}, we derive the formulas for the (pre-)Courant bracket on the adjoint tractor bundle $\cA M$ and isolate the topological obstruction to the Jacobi identity, which prevents the bracket from being a Courant bracket in general. Finally, we show that if the adjoint tractor bundle of a Cartan geometry of type $(G,H)$, where $G$ is a semisimple Lie group, admits a (pre-)Courant bracket, then it must come from a parabolic geometry. 

\subsection{Basic properties and structure}\label{subsec:basic_Courant}

We begin with the definition of a Courant algebroid, using the minimal set of axioms as given by Uchino~\cite{Uchino02} and reformulated by Kosmann-Schwarzbach~\cite{Kosmann05}.

\begin{dfn}(Kosmann-Schwarzbach~\cite{Kosmann05})\label{dfn:Courant_algebroid}
A \emph{Courant algebroid} is a vector bundle $E\to M$ over a smooth manifold $M$ equipped with a nondegenerate metric $(-,-)$, a map of vector bundles $\alpha \colon E \to TM$, called the \emph{anchor}, and a $\RR$-bilinear operation $[-,-]$, called the \emph{Courant bracket}, which satisfies the following conditions for any $f\in \smooth(M)$ and any $e_1,e_2,e_3\in E$:
\begin{equation}\label{eq:Courant_Jacobi_identity}
[e_1,[e_2 ,e_3]]=[[e_1,e_2],e_3] + [e_2, [e_1,e_3]]
\end{equation}
\begin{equation}\label{eq:non-skew-symmetry}
\alpha(e_1)(e_2,e_3) = (e_1, [e_2, e_3]) + (e_1, [e_3, e_2])
\end{equation}
\begin{equation}\label{eq:metric_preserving_property}
\alpha(e_1)(e_2,e_3) = ([e_1, e_2],e_3) + (e_2,[e_1, e_3])
\end{equation}
A Courant algebroid is said to be \emph{transitive} when the anchor $\alpha$ is surjective.
\end{dfn}

\begin{rem}
When the Jacobi identity \eqref{eq:Courant_Jacobi_identity} does not hold, we shall call the bracket a \emph{pre-Courant bracket} and the corresponding structure a \emph{pre-Courant algebroid}, following Vaisman~\cite{Vaisman05}. (It is called a twisted Courant algebroid by Sheng and Liu~\cite{ShengLiu10}.) It is then known~\cite[Lemma 3.8]{Bressler07} that the Jacobiator
\[\cJ(e_1,e_2,e_3) = [e_1,[e_2 ,e_3]] - [[e_1,e_2],e_3] - [e_2, [e_1,e_3]]\]
is actually a 1-form.
\end{rem}

We state the following properties, which follow directly from the above axioms, but were originally included as the defining axioms for a Courant algebroid. 

\begin{pro}\cite[Theorem 1]{Kosmann05}\label{Leibniz_anchor}
Let $(E\to M, (-,-),\alpha,[-,-])$ be a Courant algebroid. For all $e_1,e_2\in \Gamma(E)$ and all $f\in \smooth(M)$, the following identities hold:
\begin{enumerate}
	\item $[e_1, fe_2] = f[e_1, e_2] + (\alpha(e_1)f)e_2$,
	\item $\alpha([e_1, e_2]) = [\alpha(e_1), \alpha(e_2)]_{TM}$, 
\end{enumerate}
where $[-,-]_{TM}$ is the Lie bracket of vector fields on $M$.
\end{pro}

\begin{rem}\label{rem:true4pre}
Proposition \ref{Leibniz_anchor} is also true for pre-Courant (or twisted Courant) algebroids, and this was shown explicitly by Sheng and Liu~\cite[Lemmata 5.4--5.5]{ShengLiu10}. 
\end{rem}

We now define the $\RR$-linear map $\partial\colon \smooth(M) \to \Gamma(E)$ 
\begin{equation}\label{eq:del_metric_dfn}
(\partial f, e) = \alpha(e)f,
\end{equation}
which can also be written, equivalently, as
\begin{equation}\label{eq:del_musical_dfn}
\partial f = \sharp \circ \alpha^* df,
\end{equation}
where we have used the musical isomorphism and $\alpha^*\colon T^*M \to E^*$. It is clear that $\im\partial$ generates $\sharp \circ \alpha^*(T^*M)$, considered as a sub-bundle in $E$. By viewing $(e_2,e_3)$ as a function in $\smooth(M)$, we can use \eqref{eq:del_metric_dfn} to rewrite \eqref{eq:non-skew-symmetry} as
\[(e_1, \partial (e_2,e_3)) = (e_1, [e_2, e_3] + [e_3,e_2]).\]
Since the metric is non-degenerate, this gives us an identity for the symmetrisation of the Courant bracket
\begin{equation}\label{eq:symmetrization}
\partial (e_1,e_2) = [e_1,e_2] + [e_2,e_1].
\end{equation}
Writing the skew-symmetrisation of the Courant bracket as
\[\llbracket e_1,e_2 \rrbracket = \frac{1}{2}([e_1,e_2] - [e_2,e_1]),\]
we find an identity for the deviation of the Courant bracket from skew-symmetry:
\begin{equation}\label{eq:skew_defect}
[e_1, e_2] = \llbracket e_1,e_2 \rrbracket + \frac{1}{2}\partial (e_1,e_2).
\end{equation}

The following proposition reveals an important piece of information about the structure of a pre-Courant algebroid. We give a proof, partly along the lines of Uchino~\cite{Uchino02} and Vaisman~\cite{Vaisman05}. Our proposition is valid, however, for pre-Courant algebroids, since the proof does not invoke the Jacobi identity and, as noted in Remark \ref{rem:true4pre}, Proposition \ref{Leibniz_anchor} is still valid for pre-Courant algebroids.

\begin{pro}\label{coisotropic_kernel}
Let $(E\to M, \alpha, (-,-), [-,-])$ be a pre-Courant algebroid and the map $\partial\colon \smooth(M) \to \Gamma(E)$ be defined as in \eqref{eq:del_metric_dfn}. Then for any $f\in\smooth(M)$, 
\[\alpha(\partial f)=0,\]
i.e. $\im\partial \subseteq \ker \alpha$. Moreover, $\im\partial$ is isotropic in $E$ and $(\im\partial)^{\perp} = \ker\alpha$, so $\ker\alpha$ is coisotropic, and we have the inclusions
\[(\ker\alpha)^{\perp} \subseteq \ker\alpha \subseteq E\]
\end{pro}

\begin{proof}
Let $e_1,e_2\in\Gamma(E)$ and $f\in\smooth(M)$. We first consider the expression $[fe_1,e_2]$. Using Proposition \ref{Leibniz_anchor}(1), as well as \eqref{eq:symmetrization} twice, and observing that $\partial$ obeys the Leibniz rule on functions, we see that
\begin{equation}\label{eq:Leibniz_on_first}
[fe_1, e_2] = f[e_1, e_2] - (\alpha(e_2)f)e_1+ (e_1,e_2)\partial f.
\end{equation}
Applying the anchor to both sides of the equation and Proposition \ref{Leibniz_anchor}(2), we find that on the left-hand side, the Leibniz rule for vector fields gives
\[
[f\alpha(e_1),\alpha(e_2)]_{TM} = f[\alpha(e_1),\alpha(e_2)]_{TM} - (\alpha(e_2)f)\alpha(e_1)
\]
while on the right-hand side,
\[[f\alpha(e_1),\alpha(e_2)]_{TM} = f[\alpha(e_1),\alpha(e_2)]_{TM} - (\alpha(e_2)f)\alpha(e_1) + (e_1,e_2)\alpha(\partial f).
\]
Since $e_1,e_2$ and $f$ are arbitrary and the metric is non-degenerate, we can assume that $(e_1,e_2)$ is nonzero. Hence, $\alpha(\partial f)=0$, so $\im\partial \subseteq \ker \alpha$.

Note that since $\alpha(\partial f)=0$ is equivalent to the fact that, for any $f,g\in\smooth(M)$,
\[(\partial f, \partial g) = (\alpha(\partial g))f = 0,\]
this shows that $\im \partial \subseteq (\im\partial)^{\perp}$.

To show that $(\im\partial)^{\perp} = \ker\alpha$, observe that $TM\nsubseteq (\im\partial)^{\perp}$, so $(\im\partial)^{\perp} \subseteq \ker\alpha$. Also, for any $e\in\ker\alpha$ and $f\in\smooth(M)$, we have
\[(\partial f, e) = (\alpha(e))f = 0,\]
so $(\im\partial)^{\perp} = \ker\alpha$. Taking complements gives $\im\partial = (\ker\alpha)^{\perp} \subseteq \ker\alpha = (\im\partial)^{\perp}$, so $\ker\alpha$ is coisotropic.
\end{proof}

\begin{cor}
Let $\cE := \ker \alpha / (\ker \alpha)^{\perp}$. The following are exact sequences of vector bundles:
\[0 \to \ker \alpha \to E \to TM \to 0\]
\[0 \to (\ker \alpha)^{\perp} \to \ker \alpha \to \cE \to 0\]
Hence, $E$ splits as \[E \cong TM \oplus \cE \oplus T^*M.\]
\end{cor}

\begin{proof}
Since we are working over the smooth category, every exact sequence of vector bundles splits. The reader is referred to \cite[Lemma 1.2]{ChenStienonXu09} for details of the construction of such a splitting.
\end{proof}

\begin{rem}
Note that in the case of exact Courant algebroids, $\cE$ vanishes, since $\ker\alpha = (\ker\alpha)^{\perp} = T^*M$.
\end{rem}

The next definition will be useful in the sequel. 

\begin{dfn}\label{dfn:reduced_Lie_algebroid}
The \emph{reduced Lie algebroid} of $E$ is defined to be
\[\cA' := E/(\ker\alpha)^{\perp} \cong TM \oplus \cE.\] 
\end{dfn}

\begin{rem}\label{rem:reduced_Lie_algebroid}
Bressler has shown that $\cA'$ is actually a Lie algebroid~\cite[Lemma 2.2]{Bressler07}. To see this, observe from \eqref{eq:skew_defect} that the defect in the skew-symmetry of the Courant bracket takes values in $(\ker\alpha)^{\perp}$. Thus, the Courant bracket on $E$ restricts to a Lie algebroid bracket on $\cA'$.
\end{rem}

\subsection{The Courant bracket and the first Pontrjagin class}\label{subsec:Courant_bracket}

For a transitive pre-Courant algebroid $E\to M$, we have seen in the above that $E\cong TM \oplus \cE \oplus T^*M$, where $\cE$ is another vector bundle over $M$. We have also shown in \S\ref{subsec:Weyl_structures} that the adjoint tractor bundle $\cA M$ decomposes into the direct sum $\cA M \cong TM \oplus (\cG_0\times_{G_0} \fg_0) \oplus T^*M$ of vector bundles. In the sequel, we set 
\[\cE := \cG_0\times_{G_0} \fg_0.\]

\begin{rem}
We begin with some remarks on notation. We shall also write 
\[\cA' = TM\oplus \cE\] 
for the reduced Lie algebroid of $\cA M$, as in Definition \ref{dfn:reduced_Lie_algebroid}. 

We shall need to distinguish between various brackets of Lie and Courant algebroids. The Lie bracket on $\Gamma(\cA')$ (resp. $\Gamma(\cE)$) will be denoted by $[-,-]_{\cA'}$ (resp. $[-,-]_{\cE}$), the Courant bracket on $\Gamma(\cA M)$ by $[-,-]_{\cA M}$ and the usual Lie bracket for $\Gamma(TM)$ by $[-,-]$.

Besides the metric $(-,-)$ on the Courant algebroid $\cA M$, which is induced by the Killing form on $\fg$, we shall also encounter the natural pairing $\langle -,- \rangle$ of vector fields and 1-forms, which we shall consider to be the restriction of $(-,-)$ to $T^*M\times TM$. Since the metric on $\cA M$ is also non-degenerate when restricted to $\cE$, we shall denote this restriction by $(-,-)_{\cE}$ where necessary.
\end{rem}

From the axioms defining a Courant algebroid, we shall now derive the formulas for the Courant bracket on $\cA M$, together with any conditions that must hold for the bracket to be Courant. We first consider the cases where the Courant bracket on $\cA M$ takes at least one argument in $\Gamma(T^*M)$. These can be deduced from the properties of Courant algebroids, given in \S\ref{subsec:basic_Courant}, without any further data.

\begin{lem}\cite[Proposition 1.2(c)]{Vaisman05}
For any $e\in\Gamma(\cA M)$ and any $f\in\smooth(M)$, we have the following identities:
\begin{equation}\label{eq:Courant_for_del}
[e,\partial f] = \partial(\alpha(e)f), \quad [\partial f , e] = 0.
\end{equation}
\end{lem}

\begin{proof}
Let $e_1,e_2,e_3\in\Gamma(\cA M)$. Recall that the Jacobiator of $[-,-]$ is defined by 
\[\cJ(e_1,e_2,e_3) = [e_1, [e_2, e_3]] - [[e_1,e_2],e_3] - [e_2,[e_1,e_3]].\]
Observe that
\begin{align*}
\cJ(e_1,e_2,e_3) + \cJ(e_2,e_1,e_3) &= -[([e_1,e_2]+[e_2,e_1]),e_3]\\
&= -[\partial(e_1,e_2),e_3],
\end{align*}
where the last line follows from \eqref{eq:symmetrization}. Note also that any $f\in\smooth(M)$ may be expressed locally in terms of the metric as
\[f = \frac{1}{(e,e)}(e, fe), \quad e\in \Gamma(\cA M), (e,e)\neq 0.\]
Hence, since the Jacobiator is trivial by definition of the Courant bracket, we have the second part of \eqref{eq:Courant_for_del}. The first part follows from substituting the second part into the symmetrisation \eqref{eq:symmetrization} of the Courant bracket.
\end{proof}

\begin{cor}\label{formula_forms}
Let $X_1,X_2\in\Gamma(TM)$, $s_1,s_2\in\Gamma(\cE)$ and $\eta_1,\eta_2\in\Gamma(T^*M)$. The following identities hold:
\[[X_1, \eta_2]_{\cA M} = \cL_{X_1}\eta_2,\]
\[[\eta_1, X_2]_{\cA M} = -\cL_{X_2}\eta_1 + d(\eta_1, X_2),\]
\[[\eta_1, s_2]_{\cA M} = [s_1, \eta_2]_{\cA M} = [\eta_1, \eta_2]_{\cA M} = 0.\]
In the above, $\cL$ is the Lie derivative, $d$ is the exterior derivative and $(-,-)$ is the metric on $\cA M$.
\end{cor}

\begin{proof}
All the identities follow from \eqref{eq:Courant_for_del} in the following ways: the first identity on applying Proposition \ref{Leibniz_anchor}(1), the second identity by using \eqref{eq:Leibniz_on_first} and the last set of identities by observing that all the terms take values in $\ker\alpha$ and $T^*M$ is isotropic in $\ker\alpha$.
\end{proof}

The next step in the derivation of the formulas is to introduce a connection for the Courant algebroid. We recall Bressler's definition of a connection for a Courant algebroid.

\begin{dfn}\cite[Definition 2.5]{Bressler07}
A \emph{connection} $\gamma$ on a transitive Courant algebroid is an isotropic section of the anchor map $\alpha$.
\end{dfn}

\begin{rem}
This is not as restrictive a condition as it may seem, for we can always find an isotropic section, given any section, using the construction in Proposition \ref{existence_isotropic_splitting}. In fact, the connections that shall concern us are actually principal connections, so Proposition \ref{existence_isotropic_splitting} would be especially useful.
\end{rem}

\begin{dfn}\label{dfn:Lie_curvature}
The curvature $R^{\gamma}$ of a connection $\gamma$ on a transitive Courant (or Lie) algebroid is defined by
\[R^{\gamma}(X_1,X_2) := [\gamma(X_1),\gamma(X_2)] - \gamma([X_1,X_2]), \quad X_1,X_2\in \Gamma(TM).\]
\end{dfn}

\begin{rem}
Mackenzie~\cite{Mackenzie05} has shown that this definition of the curvature on a Lie algebroid is equivalent to the usual definition. It is easy to check that $R^{\gamma}$ takes values in $\ker\alpha$ and is actually a map $R^{\gamma}\colon \wedge^2 TM \to \ker\alpha$ or a ($\ker\alpha$)-valued 2-form on $M$.
\end{rem}

\begin{rem}
In the sequel, we shall write $\gamma:= \sigma^*\omega_0$ for the Weyl connection associated to the Weyl structure $\sigma$ of $\cA M$ (Definition \ref{dfn:Weyl_connection}). We note that $\gamma$ is automatically isotropic as a connection of $\cA'$.
\end{rem}

The next proposition gives the Courant bracket when both arguments are elements in $\Gamma(TM)$.

\begin{pro}\label{formula_vectors}
Let $\cA M$ be the adjoint tractor bundle of a parabolic geometry $(\pi\colon \cP \to M,\omega)$ with a Weyl structure $\sigma\colon \cG_0 \to \cP$ and $\gamma$ be the Weyl connection associated to $\sigma$. Then, for any $X_1,X_2\in \Gamma(TM)$,
\[[X_1,X_2]_{\cA M} = [X_1,X_2] + R^{\gamma}(X_1,X_2) + \cH(X_1,X_2,-).\]
Here, $\cH\in\Omega^3(M)$ is defined, for any $X_3\in \Gamma(TM)$, by the formula
\[\langle \cH(X_1,X_2,-),\hat{\gamma}(X_3) \rangle = ([\hat{\gamma}(X_1),\hat{\gamma}(X_2)]_{\cA M},\hat{\gamma}(X_3)),\]
where $\hat{\gamma}$ is a lift of $\gamma$ to a connection on $\cA M$, $(-,-)$ is the metric on $\cA M$, and $\langle -,- \rangle$ is the natural pairing between vector fields and 1-forms.
\end{pro}

\begin{proof}
The first term follows from Proposition \ref{Leibniz_anchor}(2). For the second term, we observe that the first two terms should give the bracket $[X_1,X_2]_{\cA'}$ of $X_1$ and $X_2$ on $\cA'$. Since we have shown in Corollary \ref{Atiyah_of_G0} that the Weyl connection is in fact a connection of the Atiyah algebroid of $\cG_0$, which is actually the reduced Lie algebroid $\cA'$, the second term can be deduced from Definition \ref{dfn:Lie_curvature}.

Since $R^{\hat{\gamma}}$, the curvature of $\hat{\gamma}$, takes values in $\ker\alpha = \cE \oplus T^*M$, we note that it splits into
\[R^{\hat{\gamma}} = \widehat{R^{\gamma}} + \cH^{\hat{\gamma}}\]
where $\widehat{R^{\gamma}}\in\Omega^2(M, \cE)$ is the lift of $R^{\gamma}$ to $\cA M$ and $\cH^{\hat{\gamma}}\in \Omega^2(M) \otimes T^*M$. The calculations in \cite[Lemma 3.6]{Bressler07} then show that $\cH$ is indeed totally skew-symmetric, hence in $\Omega^3(M)$, and satisfies the given formula.
\end{proof}

\begin{pro}\label{fomulas_endo}
With the same notation as in Proposition \ref{formula_vectors}, and for any $X_1\in \Gamma(TM)$ and $s_1,s_2\in \Gamma(\cE)$, we have
\[[X_1,s_2]_{\cA M} = -[s_2,X_1]_{\cA M} = [\gamma(X_1),s_2]_{\cA'} - (R^{\gamma}(X_1,-),s_2)\]
\[[s_1,s_2]_{\cA M} = [s_1,s_2]_{\cE} + ([\gamma(-),s_1]_{\cA'},s_2)\]
\end{pro}

\begin{proof}
It is clear from Proposition \ref{Leibniz_anchor}(2) that $[X_1,s_2]_{\cA M}$ and $[s_1,s_2]_{\cA M}$ cannot have any terms in $\Gamma(TM)$, so the first terms of each identity are as shown above. It remains to show that the 1-form is defined as such for each identity. Suppose $\phi(X_1,s_2)$ is the 1-form in $[X_1,s_2]_{\cA M}$, and let $X_3\in \Gamma(TM)$. Then we have
\begin{align*}
\langle\phi(X_1,s_2),X_3\rangle &= (\phi(X_1,s_2),\hat{\gamma}(X_3))\\
&= ([\hat{\gamma}(X_1),s_2]_{\cA M}, \hat{\gamma}(X_3))\\
&= (X_1)(\hat{\gamma}(X_3),s_2) - ([\hat{\gamma}(X_1),\hat{\gamma}(X_3)]_{\cA M},s_2)\\
&= -(\hat{\gamma}([X_1,X_3]) + \widehat{R^{\gamma}(X_1,X_3)} + \cH(X_1,X_3,-),s_2)\\
&= -(\widehat{R^{\gamma}(X_1,X_3)},s_2),
\end{align*}
as required. Since $(X_1,s_2) = 0$, it follows from \eqref{eq:symmetrization} that
\[[X_1,s_2]_{\cA M} + [s_2,X_1]_{\cA M} = 0,\]
so $[X_1,s_2]_{\cA M} = -[s_2,X_1]_{\cA M}$. For the second identity, let $\varphi(s_1,s_2)$ be the 1-form. It follows that
\begin{align*}
\langle \varphi(s_1,s_2),X_3 \rangle &= (\varphi(s_1,s_2),\hat{\gamma}(X_3))\\
&= ([s_1,s_2]_{\cA M},\hat{\gamma}(X_3))\\
&= -(s_2, [s_1,\hat{\gamma}(X_3)]_{\cA M})\\
&= ([\hat{\gamma}(X_3),s_1]_{\cA M}, s_2)\\
&= ([\gamma(X_3),s_1]_{\cA'}, s_2),
\end{align*}
as required. This completes the proof.
\end{proof}

Finally, we need to find the conditions that guarantee the Jacobi identity for the bracket defined by the formulas. We first make this observation.

\begin{lem}
The map $\nabla \colon \Gamma(TM) \otimes \Gamma(\cE) \to \Gamma(\cE)$, defined for any $X\in\Gamma(TM)$ and $s\in\Gamma(\cE)$ by
\[\nabla_X s := [\gamma(X),s]_{\cA'},\]
is a Koszul connection on $\cE$, so it satisfies, for any $f\in\smooth(M)$, the equations
\[\nabla_{(fX)} s = f\nabla_X s \quad \textrm{and} \quad \nabla_X(fs) = f\nabla_X s + X(f)s\]
\end{lem}

\begin{proof}
This is a consequence of the Leibniz rule for Lie algebroids and the fact that the projection of $s$ under the anchor vanishes.
\end{proof}

\begin{rem}
To simplify notation, we shall write, for any $X\in\Gamma(TM)$ and $s\in\Gamma(\cE)$, the expression $[\gamma(X),s]_{\cA'}$ as $\nabla_X s$.
\end{rem}

Next, we recall the following definition of the first Pontrjagin class.

\begin{dfn}(Bressler~\cite{Bressler07}, Chen et al.~\cite{ChenStienonXu09})\label{dfn:Pontrjagin}
Let $\cA$ be a transitive Lie algebroid with bracket $[-,-]_{\cA}$, anchor $\alpha$ and $\cE := \ker \alpha$. Suppose $\cA$ is endowed with an $\ad$-invariant fiber-wise non-degenerate bilinear pairing $(-,-)_{\cE}$ on $\cE$ such that the bracket is compatible with the pairing. Then $(\cA, (-,-)_{\cE})$ is called a \emph{quadratic} Lie algebroid by Chen et al.~\cite{ChenStienonXu09}.

Let $\gamma$ be a section of the anchor $\alpha$ and $R^{\gamma}$ be its curvature (Definition \ref{dfn:Lie_curvature}). Then the cohomology class of the 4-form $\langle R^{\gamma} \wedge R^{\gamma} \rangle \in \Omega^4(M)$, which is defined for any $X_1,X_2,X_3,X_4 \in \Gamma(TM)$ by
\[\langle R^{\gamma} \wedge R^{\gamma} \rangle (X_1,X_2,X_3,X_4) := \frac{1}{4} \sum_{\tau\in S_4} (R^{\gamma}(X_{\tau(1)}, X_{\tau(2)}), R^{\gamma}(X_{\tau(3)}, X_{\tau(4)}))_{\cE}\]
is called the \emph{first Pontrjagin class} associated to $(\cA, (-,-)_{\cE})$.
\end{dfn}

\begin{rem}
Note that $\langle R^{\gamma} \wedge R^{\gamma} \rangle$ is a closed 4-form and the first Pontrjagin class does not depend on $\gamma$. We also remark that the reduced Lie algebroid $(\cA', (-,-)_{\cE})$ of $\cA M$ satisfies the conditions of Definition \ref{dfn:Pontrjagin}, so it has an associated first Pontrjagin class.
\end{rem}

We are now ready to state the conditions that must hold for the Jacobi identity to be true for the bracket just defined.

\begin{pro}
Let $X_1,X_2,X_3\in \Gamma(TM)$ and $s_1,s_2,s_3\in\Gamma(\cE)$. For the formulas given in Corollary \ref{formula_forms} and Propositions \ref{formula_vectors} and \ref{fomulas_endo} to define a Courant bracket on $\cA M$, the following compatibility equations must be satisfied:
\[\nabla_{X_1} [s_2,s_3]_{\cE} = [\nabla_{X_1} s_2, s_3]_{\cA'} + [s_2, \nabla_{X_1} s_3]_{\cA'};\]
\[\nabla_{X_1} \nabla_{X_2} s_3 - \nabla_{X_2} \nabla_{X_1} s_3 - \nabla_{[X_1,X_2]} s_3 =[R^{\gamma}(X_1,X_2),s_3]_{\cE};\]
\[\set{\nabla_{X_1} R^{\gamma}(X_2,X_3) + R^{\gamma}([X_1,X_2],X_3)} + \textrm{cyclic permutations} = 0;\]
the condition that
\[d\cH = \frac{1}{2} \langle R^{\gamma} \wedge R^{\gamma} \rangle,\]
where $\cH$ is the 3-form defined in Proposition \ref{formula_vectors} and $\langle R^{\gamma} \wedge R^{\gamma} \rangle$ is the closed 4-form defined in Definition \ref{dfn:Pontrjagin}; and the Cartan relations.
\end{pro}

\begin{proof}
It is clear from Corollary \ref{formula_forms} that the Jacobi identity is trivially satisfied when it involves more than one 1-form, or one 1-form and one element from $\Gamma(\cE)$. For the cases where the Jacobi identity involves two vector fields and one 1-form, the Cartan relations ensure that the Jacobi identity is satisfied. It remains to consider the remaining cases.

For the case
\[[X_1, [s_2,s_3]_{\cA M} ]_{\cA M} = [[X_1,s_2]_{\cA M}, s_3]_{\cA M} + [s_2, [X_1, s_3]_{\cA M} ]_{\cA M},\]
it is easy to see that this reduces to the equation
\[\nabla_{X_1} [s_2,s_3]_{\cE} = [\nabla_{X_1} s_2, s_3]_{\cA'} + [s_2, \nabla_{X_1} s_3]_{\cA'}.\]
In the case where
\[[X_1, [X_2,s_3]_{\cA M} ]_{\cA M} = [[X_1,X_2]_{\cA M}, s_3]_{\cA M} + [X_2, [X_1, s_3]_{\cA M} ]_{\cA M},\]
we see that
\begin{align*}
[X_1, [X_2,s_3]_{\cA M} ]_{\cA M} &= [X_1, \nabla_{X_2} s_3 - (R(X_2, -),s_3)_{\cE}]_{\cA M}\\
&= \nabla_{X_1}\nabla_{X_2} s_3
\end{align*}
and, similarly, $[X_2, [X_1, s_3]_{\cA M} ]_{\cA M} = \nabla_{X_2}\nabla_{X_1} s_3$. Since
\[[[X_1,X_2]_{\cA M}, s_3]_{\cA M} = \nabla_{[X_1,X_2]} s_3 + [R^{\gamma}(X_1,X_2), s_3]_{\cE},\]
we see that the Jacobi identity in this case is equivalent to
\[\nabla_{X_1} \nabla_{X_2} s_3 - \nabla_{X_2} \nabla_{X_1} s_3 - \nabla_{[X_1,X_2]} s_3 =[R^{\gamma}(X_1,X_2),s_3]_{\cE}.\]
The last two identities, together with the Jacobi identity for vector fields on $M$, are equivalent to the Jacobi identity
\[[X_1, [X_2,X_3]_{\cA M} ]_{\cA M} = [[X_1,X_2]_{\cA M}, X_3]_{\cA M} + [X_2, [X_1, X_3]_{\cA M} ]_{\cA M}.\]
The verification is a long but straightforward calculation, and we refer to \cite[Lemma 3.8]{Bressler07} for the full details of that calculation.
\end{proof}

We collect the formulas for the Courant bracket in the Theorem below.

\begin{thm}\label{main_theorem}
Let $(\cP\to M, \omega)$ be a parabolic geometry of type $(G,P)$ over a smooth manifold $M$ with a Cartan connection $\omega$. Suppose there is a Weyl structure $\sigma$ on $\cP$, so that we have a Weyl connection $\gamma$, which is also a principal connection on $\cG_0 := \cP / P_+$, with curvature $R^{\gamma}$. Let 
\[\cA M\cong TM \oplus \cE \oplus T^*M\] 
be the adjoint tractor bundle of $\cP$ with an $\ad$-invariant metric $(-,-)$ induced from the Killing form on $\fg$. For any $X_1,X_2,X_3\in \Gamma(TM)$, $s_1,s_2,s_3\in\Gamma(\cE)$ and $\eta_1,\eta_2\in\Omega^1(M)$, the formulas
\begin{align*}
[X_1,X_2]_{\cA M} &= [X_1,X_2] + R^{\gamma}(X_1,X_2) + \cH(X_1,X_2,-),\\
[X_1,s_2]_{\cA M} &= [\gamma(X_1),s_2]_{\cA'} - (R^{\gamma}(X_1,-),s_2) = -[s_2,X_1]_{\cA M},\\
[s_1,s_2]_{\cA M} &= [s_1,s_2]_{\cE} + ([\gamma(-),s_1]_{\cA'},s_2),\\
[X_1, \eta_2]_{\cA M} &= \cL_{X_1}\eta_2,\\
[\eta_1, X_2]_{\cA M} &= -\cL_{X_2}\eta_1 + d(\eta_1, X_2),\\
[\eta_1, s_2]_{\cA M} &= [s_1, \eta_2]_{\cA M} = [\eta_1, \eta_2]_{\cA M} = 0,
\end{align*}
together with the compatibility equations
\[\nabla_{X_1} [s_2,s_3]_{\cE} = [\nabla_{X_1} s_2, s_3]_{\cA'} + [s_2, \nabla_{X_1} s_3]_{\cA'},\]
\[\nabla_{X_1} \nabla_{X_2} s_3 - \nabla_{X_2} \nabla_{X_1} s_3 - \nabla_{[X_1,X_2]} s_3 =[R^{\gamma}(X_1,X_2),s_3]_{\cE},\]
\[\set{\nabla_{X_1} R^{\gamma}(X_2,X_3) + R^{\gamma}([X_1,X_2],X_3)} + \textrm{cyclic permutations} = 0,\]
the topological condition
\begin{equation}\label{eq:exact_Pontrjagin}
d\cH = \frac{1}{2} \langle R^{\gamma} \wedge R^{\gamma} \rangle,
\end{equation}
and the Cartan relations, define a Courant bracket $[-,-]_{\cA M}$ on $\cA M$.
\end{thm}

\begin{rem}
Condition \eqref{eq:exact_Pontrjagin} states that $\langle R^{\gamma} \wedge R^{\gamma} \rangle$ is an exact 4-form. In other words, the first Pontrjagin class of $(\cA', (-,-)_{\cE})$ must vanish for the bracket $[-,-]_{\cA M}$ to define a Courant bracket on $\cA M$. 

If the first Pontrjagin class does not vanish, the adjoint tractor bundle then becomes what Hansen and Strobl~\cite[Definition 1]{HansenStrobl10} call an $H_4$-twisted Courant algebroid, where $H_4$ is the closed 4-form given by
\[H_4(X_1,X_2,X_3,-) = (-\frac{1}{2}\langle R^{\gamma} \wedge R^{\gamma} \rangle + dH)(X_1,X_2,X_3,-).\]
This is equivalent to the notion of a pre-Courant algebroid and it has been shown by Sheng and Liu that this gives rise to what is called a Leibniz 2-algebra in~\cite[Theorem 5.3]{ShengLiu10}.

Classes in degree 4 cohomology can also be associated to bundle 2-gerbes, as shown by Carey et al.~\cite{CJMSW05} and Stevenson~\cite{Stevenson04}, but we shall postpone the investigation of these objects to future work.
\end{rem}

We also get an interesting characterisation of parabolic geometries with flat Weyl connections.

\begin{cor}
Let $(\cP\to M, \omega)$ be a parabolic geometry of type $(G,P)$ over a smooth manifold $M$ with a Cartan connection $\omega$. Let $\sigma$ be a Weyl structure on $\cP$ and $\gamma$ be the associated Weyl connection. If $\gamma$ is flat, then the bracket $[-,-]_{\cA M}$ defined in Theorem \ref{main_theorem} endows the adjoint tractor bundle $\cA M$ of the parabolic geometry with the structure of a Courant algebroid.
\end{cor}

\subsection{(Pre-)Courant structures on Atiyah algebroids}\label{subsec:Courant_Atiyah}

In the following, we shall write $(\fg,\fh)$ for a pair of Lie algebras, where $\fh$ is a Lie subalgebra of $\fg$ and $\fg$ is assumed to be semisimple, and let $(G,H)$ denote the corresponding pair of connected Lie groups. 

We shall show that if the Atiyah algebroid of a principal $H$-bundle $\cP$ over a ($\dim\fg -\dim\fh$)-dimensional manifold $M$, with model fibre $\fg$, admit the structure of a pre-Courant algebroid, then $\fh$ is a parabolic subalgebra of $\fg$. Equivalently, this shows that for any Cartan geometry of type $(G,H)$ with $G$ semisimple, the adjoint tractor bundle can be equipped with a pre-Courant algebroid structure if and only if $H$ is a parabolic subgroup of $G$.

It is clear from the $|k|$-grading on $\fg$ (Definition \ref{dfn:k-grading}) induced by the choice of a parabolic subalgebra that parabolic subalgebras of semisimple Lie algebras are coisotropic. Our desired result follows from the converse, which turns out to be an old result of Dixmier~\cite[Lemme 1.1(i)]{Dixmier76}, who attributed the proof to P. Tauvel. We shall require Dixmier's result in a form (Proposition \ref{coisotropic_parabolic}) somewhat different from the version in the literature.

Dixmier's result requires the following observation. For any $h\in \fg$ with $\eta=B(h,-)\in \fg^*$, there is an associated skew-symmetric bilinear form given by
\[B_h(X,Y) = B(h,[X,Y]) = \eta([X,Y]), \quad \forall\, X,Y\in\fg.\]
Suppose $\fp$ is any subalgebra of $\fg$. Then we write
\[\fp^{\perp_h} := \set{X\in\fg\: | \: B(h,[X,Y])=0 \quad \forall\, Y\in\fp}.\]
We note that $h\in\fg^{\perp_h} \subset \fp^{\perp_h}$ by the invariance of the Killing form.

\begin{rem}
The notion of coisotropy used by Dixmier~\cite{Dixmier76} in the hypothesis for his Lemme 1.1 is that $\fp^{\perp_h} \subset \fp$. Similar notions have arisen, for example, in Poisson geometry and Lie bialgebras (see Li-Bland and Meinrenken~\cite{LiBlandMeinrenken09}, Weinstein~\cite{Weinstein88}, Zambon~\cite{Zambon11}). Our notion of coisotropy is somewhat different, but produces the desired result. For completeness, we give details of the proof, which closely follow the original, since it does not seem to be well-known.
\end{rem}

\begin{pro}\label{coisotropic_parabolic}
Let $\fg$ be a complex semisimple Lie algebra and $\fp$ be a subalgebra of $\fg$ such that $\fp^{\perp} \subset \fp$. Then $\fp$ is a parabolic subalgebra of $\fg$.
\end{pro}

\begin{proof}
We use the invariance of the Killing form to see that
\[B([h,\fp^{\perp_h}],\fp) = B(h,[\fp^{\perp_h},\fp]) = 0 \]
This shows that $[h,\fp^{\perp_h}]\subset \fp^{\perp}$. Swapping $\fp^{\perp_h}$ and $\fp$ in the bracket, we find that $(\fp^{\perp_h})^{\perp_h}= \fp$, so $h\in\fp$ as well.

To see the reverse inclusion, note that the map $\ad(h)\colon \fp^{\perp_h} \to [h, \fp^{\perp_h}]$ vanishes precisely over $\fg^{\perp_h}$. It follows that $\dim ([h,\fp^{\perp_h}]) = \dim (\fp^{\perp})$, so we have 
\[[h,\fp^{\perp_h}] = \fp^{\perp}.\]
It follows from the invariance of the Killing form that $[\fp^{\perp},\fp] \subset \fp^{\perp}$, so $\fp^{\perp}$ is an ideal in $\fp$, and hence a subalgebra of $\fg$. 

Since $\fg$ is a complex semisimple Lie algebra, we can choose a Cartan subalgebra $\fh_0$ and a root system $\Delta$ of $\fg$ with a subsystem of positive roots $\Delta^+$. Note that if
\[\fh_+ := \bigoplus_{\eta\in\Delta^+} \fg_{\eta},\]
then the Killing form is degenerate on $\fh_+$ and $\fb = \fh_0\oplus \fh_+$ is a Borel subalgebra of $\fg$. It follows that $[h, \fp^{\perp_h}] = \fp^{\perp}\subset \fh_+$, since $\fh_+$ is maximal amongst the subalgebras in $\fg$ for which the Killing form is degenerate. Now,
\begin{align*}
B(h, [\fp^{\perp_h}, \fh_0\oplus \fh_+]) &= B([h, \fp^{\perp_h}], \fh_0\oplus \fh_+)\\
& \subset B(\fh_+, \fh_0\oplus \fh_+),
\end{align*}
but $B(\fh_+, \fh_0\oplus \fh_+) = 0$, which implies that $\fh_0\oplus \fh_+ \subset (\fp^{\perp_h})^{\perp_h}= \fp$. Since $\fp$ contains a Borel subalgebra of $\fg$, it follows that $\fp$ is a parabolic subalgebra of $\fg$.
\end{proof}

\begin{rem}
The result also holds when we work over the reals, since a parabolic subalgebra $\fp$ of a real semisimple Lie algebra $\fg$ is defined to be a subalgebra such that its complexification $\fp_{\CC}$ is a parabolic subalgebra of $\fg_{\CC}$.
\end{rem}

We now state the following theorem. The real case follows from complexification.

\begin{thm}\label{classification}
Let $\fg$ be a complex semisimple Lie algebra and $\fh\subset \fg$ be a Lie subalgebra. Let $(G,H)$ be the corresponding pair of connected Lie groups and $M$ be a manifold of dimension $n=\dim \fg - \dim \fh$. Then the Atiyah algebroid of a principal $H$-bundle $\cP$ over $M$
\[0 \to \cP\times_{H}\fh \to \cP\times_{H}\fg \to TM \to 0\]
admits the structure of a transitive (pre-)Courant algebroid if and only if $\fh$ is a parabolic subalgebra of $\fg$.
\end{thm}

\begin{proof}
We have established the converse in \S \ref{subsec:Courant_bracket}. Now recall from Proposition \ref{coisotropic_kernel} that if $E$ is a pre-Courant algebroid with
\[0 \to \ker\alpha \to E \xrightarrow{\alpha} TM \to 0,\]
then $(\ker\alpha)^{\perp}\subseteq \ker\alpha$ with respect to the metric on $E$. If $E$ is an Atiyah algebroid, this reduces to finding all subalgebras $\fh$ of $\fg \cong (\fg/\fh) \oplus \fh$ such that, with respect to the Killing form, either $\fh$ is coisotropic (and $E$ is non-exact) or $\fh^{\perp} = \fh$ (and $E$ is exact).

If $\fh^{\perp} = \fh$, the properties of the exact pre-Courant algebroid imply that $\fh$ and $\fg/\fh$ are dual to each other, i.e. $\fg/\fh \cong \fh^*$. However, since the fibres of $TM$ are supposed to be modelled by $\fg/\fh$, this implies that $\fg/\fh$, and therefore $\fh$, must be isomorphic to (a direct sum of copies of) $\CC^n$. Hence, $\fg$ cannot be semisimple, which contradicts our hypotheses.

Proposition \ref{coisotropic_parabolic} then shows that the coisotropic case is equivalent to $\fh$ being a parabolic subalgebra of $\fg$. The results in \S \ref{subsec:Courant_bracket} show that, in this case, the Atiyah algebroid admits the structure of a non-exact transitive (pre-)Courant algebroid, which concludes the proof.
\end{proof}

\begin{rem}
We note that if we do not require $\fg$ to be semisimple, then $\fh$ need only be coisotropic in $\fg$. One could conceivably consider Cartan geometries in which this condition holds (see Xu~\cite{Xu12} for example), but since parabolic geometries are already fairly abundant and extensively studied in the literature, we have chosen to focus on these geometries.
\end{rem}

\section{Concluding remarks}\label{sec:Conclusion}

In this paper, we have shown that parabolic geometries are a source of transitive (pre-)Courant algebroids which are \emph{a priori} non-exact. The first Pontrjagin class of the associated Atiyah algebroid, represented by a closed 4-form, naturally appears in the process and would be interesting to study in relation to the parabolic geometry, particularly since classes in degree 4 cohomology can also be associated to bundle 2-gerbes (see, for example, the papers by Carey et al.~\cite{CJMSW05} and Stevenson~\cite{Stevenson04}).

Just as exact Courant algebroids have played a major role in the study of T-duality, via generalised complex geometry, non-exact Courant algebroids have also begun to appear in recent works generalising T-duality. For example, recent work by Hohm and Kwak~\cite{HohmKwak11}, who were extending the double field theory of Hull and Zwiebach~\cite{HullZwiebach09a,HullZwiebach09b} to heterotic strings, has shown that the gauge algebra, in this case, turns out to be a non-exact Courant algebroid.

The notion of a Courant algebroid has also been generalised in many directions. For example, Sheng and Liu~\cite{ShengLiu10} have studied Leibniz 2-algebras, which can arise from the pre-Courant algebroids studied here. The notion of a Leibniz algebroid, which sacrifices a weak skew symmetry for a strict Jacobi identity, has been introduced by Baraglia~\cite{Baraglia11a}. Interestingly, Baraglia~\cite{Baraglia11b} and Li-Bland~\cite{Li-Bland11} have introduced more exotic generalisations of Courant algebroids to study generalised versions of parabolic geometries, which may hint at a relation with the pre-Courant algebroids studied here.

After a draft of this paper was completed and circulated, several preprints have appeared that build upon the work done here, including those by Liu et al.~\cite{LiuShengXu12} and Xu~\cite{Xu12}. We are heartened by the growing interest in the topic.

\section*{Acknowledgements}

The authors would like to thank the Erwin Schr\"odinger Institute (ESI) in Vienna, where this work was begun, for its hospitality and excellent working conditions. The second author also gratefully acknowledges the receipt of an ESI Junior Research Fellowship. 

An earlier version of this work was presented by the second author at the Higher Structures in China II conference, held at the Mathematics School and Institute of Jilin University in Changchun, China. We warmly thank the organisers for the opportunity to present this work and the hospitality during the conference.

We are deeply grateful to Pavol \v{S}evera for pointing out the work of Li-Bland and Meinrenken~\cite{LiBlandMeinrenken09} and making the observation contained in a result of Dixmier~\cite{Dixmier76}, which resulted in the addition of \S \ref{subsec:Courant_Atiyah}.

We would like to thank David Baraglia for useful discussions and comments on an earlier version of this paper. Thanks are also due to David Roberts and Raymond Vozzo for additional comments.

\bibliographystyle{amsplain}
\bibliography{workbib}

\end{document}